\documentclass[a4paper,leqno]{amsart}

\usepackage{amsmath,amssymb,amsthm}
\usepackage[colorlinks=true]{hyperref}

\def\R{{\mathbb R}}
\def\N{{\mathbb N}}

\def\Sph{{\mathbb S}} 
\def\Sch{{\mathcal S}}
\def\F{\mathcal F}

\def\O{\mathcal O}

\def\virgp{\raise 2pt\hbox{,}}

\def\({\left(}
\def\){\right)}
\def\<{\left\langle}
\def\>{\right\rangle}
\def\le{\leqslant}
\def\ge{\geqslant}

\def\Tend#1#2{\mathop{\longrightarrow}\limits_{#1\rightarrow#2}}

\def\d{{\partial}}

\def\eps{\varepsilon}
\def\l{\lambda}
\def\om{\omega}

\DeclareMathOperator{\DIV}{\mathrm{div}}

\newcommand{\D}{d}

\def\cst{\frac{4\pi}{3}}

\theoremstyle{plain}
\newtheorem{theorem}{Theorem}[section]
\newtheorem{lemma}[theorem]{Lemma}
\newtheorem{corollary}[theorem]{Corollary}
\newtheorem{proposition}[theorem]{Proposition}

\theoremstyle{definition}
\newtheorem{definition}[theorem]{Definition}

\newtheorem{remark}[theorem]{Remark}
\newtheorem*{remark*}{Remark}

\numberwithin{equation}{section}

\begin{document}

\title[Dipolar Gross--Pitaevskii equation]{On the Gross--Pitaevskii
  equation for trapped dipolar quantum gases}     
\author[R. Carles]{R\'emi Carles}
\address[R. Carles]{CNRS \& Universit\'e Montpellier~2\\Math\'ematiques
\\CC~051\\Place Eug\`ene Bataillon\\34095
  Montpellier cedex 5\\ France}
\email{Remi.Carles@math.cnrs.fr}
\author[P. Markowich]{Peter A. Markowich}
\address[P. Markowich]{Department of Applied Mathematics and
  Theoretical Physics\\ 
  CMS, Wilberforce Road\\ Cambridge CB3 0WA\\ England}
\email{p.a.markowich@damtp.cam.ac.uk}
\author[C. Sparber]{Christof Sparber}
\address[C. Sparber]{Department of Applied Mathematics and Theoretical
  Physics\\ 
  CMS, Wilberforce Road\\ Cambridge CB3 0WA\\ England}
\email{c.sparber@damtp.cam.ac.uk}
\begin{abstract}
We study the time-dependent Gross--Pitaevskii equation describing
Bose--Einstein condensation of trapped dipolar quantum gases.  
Existence and uniqueness as well as the possible blow-up of solutions
are studied. Moreover, we discuss the problem of dimension-reduction
for this nonlinear and nonlocal Schr\"odinger equation.
\end{abstract}


\thanks{This work has been supported by the KAUST Investigator Award of P. Markowich. R.C. is partially
  supported by the ANR project SCASEN. P.M. 
  acknowledges support from  the Royal Society through his ``Wolfson
Research Merit Award''. C. S. has been supported by the ``APART
grant'' of the Austrian Academy of Science.}
\maketitle

\section{Introduction}
\label{sec:intro}

The success of atomic Bose--Einstein condensation has stimulated great
interest in the properties of trapped quantum gases. 
Recent developments in the manipulation of such ultra-cold atoms have
paved the way towards Bose--Einstein condensation in  
atomic gases where dipole-dipole interactions between the particles
are important. In \cite{YiYou} Yi and You were the first to introduce  
a pseudo-potential appropriate to describe such systems in which
particles interact via short-range repulsive forces and  
long-range (partly attractive) dipolar forces.  
Describing the corresponding Bose--Einstein condensates within the
realm of the Gross--Pitaevskii (mean field) approximation, one is  
led to the following nonlinear Schr\"odinger equation for the
macroscopic wave function of the condensate 
 \cite{Goral, Ronen, SSZL, YiYou1, YiYou2}:
\begin{equation}
  \label{eq:GPE}
  i \hbar \d_t \psi + \frac{\hbar^2}{2 m}\Delta \psi = V(x)\psi + g  \lvert
  \psi\rvert^2 \psi + d^2 \(K\ast \lvert\psi\rvert^2\)\psi,\quad x \in \R^3, t >0,
\end{equation}
where $|g|= 4 \pi \hbar^2 N |a| /m$. Here, we denote by $N$ the total
number of particles within the condensate, whereas 
$m$  denotes the mass of an individual particle and 
$a$ its corresponding scattering length (which can be experimentally
tuned to be either positive or negative).  
The wave function $\psi$ is then normalized, such that $\| \psi \|_{L^2}^2 = 1$. 
The potential $V(x)$, for $x=(x_1,x_2,x_3)\in \R^3$, describes the electromagnetic trap for the
condensate and is usually chosen to be a harmonic confinement, i.e.
\begin{equation}
  \label{eq:potential}
  V(x)=\frac{m}{2} \left(\om_1^2 x_1^2+\om_2^2x_2^2+\om_3^2x_3^2\right) . 
\end{equation}
The real-valued constants $\omega_1, \omega_2, \omega_3$  then
represent the corresponding  trap frequency in each spatial direction. 
Finally, $d$ denotes the dipole moment (in Gaussian units) and 
\begin{equation}\label{dipolekernel}
  K(x) = \frac{1-3\cos^2\theta}{\lvert x\rvert^3},
\end{equation}
where $\theta$ stands for the angle between $x\in \R^3$ and the dipole
axis $n\in \R^3$, with $|n| =1$. In other words $\theta$ is defined via 
\begin{equation*}
\cos \theta= \frac{x\cdot n}{\lvert x\rvert}.
\end{equation*}
To our knowledge a rigorous mathematical study of \eqref{eq:GPE} has
not been given so far.  
Note that the interaction kernel \eqref{dipolekernel} is indeed highly
singular and it is therefore not clear \emph{a priori} if the
corresponding  convolution operator is well defined. 
In particular, not even the existence and uniqueness of solutions to
\eqref{eq:GPE} has been established yet, and it will be one of the main
tasks of this work to do so. Moreover, we shall be interested in the
mathematical problem of dimension reduction. 
Quasi two-dimensional (pancake shaped) or even quasi one-dimensional
(cigar shaped) Bose--Einstein condensates can be obtained experimentally by
appropriately tuning the trap-frequencies $\omega_1, \omega_2,
\omega_3$. The mathematical derivation of effective models in lower
dimensions via the corresponding scaling limits is therefore of great
practical importance. 

For the mathematical analysis it is more convenient to rescale
\eqref{eq:GPE} in dimensionless form (see e.g. \cite{BJM}), to arrive
at the following model 
\begin{equation}
  \label{eq:finalNLS}
  i\d_t \psi + \frac{1}{2}\Delta \psi =  V(x)\psi +\l_1 \lvert
  \psi\rvert^2 \psi + \l_2\(K\ast \lvert\psi\rvert^2\)\psi,\quad x\in \R^3,
\end{equation}
where $\lambda_1 = 4\pi a N / a_0$ and $\lambda_2 = d^2/ (\hbar
\omega_0 a^3_0)$. Here,  we denote by $a_0 = \sqrt{\hbar / m
  \omega_0}$ the ground state length of a harmonic oscillator
corresponding to $\omega_0 = \frac{1}{3}( \omega_1+ \omega_2 +
\omega_3)$. In the upcoming analysis $\l_1,\l_2$ will simply be
assumed to be two given, real-valued parameters.  

The paper is then organized as follows: We start by collecting several
properties of the interaction kernel $K$ in
Section~\ref{sec:kernel}. After that, we set up a local in time
existence theory in Section~\ref{sec:local}, before proving several
different global existence results in
Section~\ref{sec:global}. Section~\ref{sec:unstable} is devoted to the 
question of finite time blow-up of solutions and also includes an
additional global in time existence result for rather particular
circumstances. Finally, we shall study the dimensional reduction of
our model in Section~\ref{sec:reduction}.

\section{Some properties of the dipole kernel}
\label{sec:kernel}
To simplify notations, we shall from now on assume,  
without restriction of generality, that $n=(0,0,1)$. The dipole-interaction 
kernel $K$ then reads:
\begin{equation}
  \label{eq:kernel}
  K(x)= \frac{x_1^2+x_2^2-2x_3^2}{\lvert x\rvert^5}.
\end{equation}
Even though this kernel is highly singular (like $1/|x|^3$), it
defines a rather smooth operator. The technical reason for this is that the
average of $K$ vanishes on spheres.

\begin{lemma}\label{lem:CZ}
  The operator ${\mathcal K}: u\mapsto K\ast u $ can be extended as a
  continuous operator on 
  $L^p(\R^3)$ for all $1<p<\infty$.  
\end{lemma}
\begin{proof}
  We notice that $\Omega:x\mapsto 1-3\cos^2 \theta= 1-3\frac{x_3^2}{|x|^2}$
  is Lipschitzean on 
  ${\mathbb S}^2$, 
  homogeneous of degree 
  zero, and 
  \begin{equation*}
    K(x) = \frac{\Omega(x)}{\lvert x\rvert^3}
  \end{equation*}
for $x\in \R^3$. We also note that the average of $\Omega$ on spheres
vanishes:
\begin{equation*}
  \int_{{\mathbb S}^2}\Omega =0.
\end{equation*}
The lemma then follows from the Calder\'on--Zygmund Theorem (see
e.g. \cite{Stein}). 
\end{proof}
\begin{remark} The fact that the average of $K$ vanishes on $\mathbb
  S^2$ implies that the dipole-nonlinearity vanishes when applied to
  radially symmetric wave functions $\psi = \psi(|x|)$. In this case
  \eqref{eq:finalNLS} simplifies to the classical (cubic nonlinear)
  Gross--Pitaevskii equation. 
\end{remark}
The continuity at the $L^2$-level can also be seen by computing the Fourier 
transform of $K$, which turns out to be essentially bounded, i.e. $\widehat K\in
L^\infty(\R^3)$. In the following we shall use the explicit formula for this
Fourier transform several times, see also \cite{Goral,Ronen}. 
\begin{lemma}\label{lem:FourierK}
  Define the Fourier transform on the Schwartz space as
  \begin{equation*}
    \F u(\xi)\equiv \widehat u(\xi) =
    \int_{\R^3}e^{-ix\cdot \xi}u(x)\D x,\quad u\in
    \Sch(\R^3). 
  \end{equation*}
Then the Fourier transform of $K$ is given by 
\begin{equation*}
  \widehat K(\xi)= \frac{4\pi}{3}
  \(3\cos^2\Theta-1\)=\frac{4\pi}{3} \(3\frac{\xi_3^2}{\lvert
  \xi\rvert^2}-1\)  =
  \frac{4\pi}{3} \( \frac{2\xi_3^2-\xi_1^2-\xi_2^2}{\lvert
  \xi\rvert^2}\),  
\end{equation*}
where $\Theta$ stands for the angle between $\xi$ and the dipole
axis $n=(0,0,1)$. 
\end{lemma}
\begin{proof}
  We use the decomposition of $e^{-ix\cdot \xi}$ into spherical harmonics
  \begin{equation*}
    e^{-ix\cdot \xi} = \sum_{m=0}^\infty (-i)^m (2m+1)j_m
    \(\lvert x\rvert \lvert \xi\rvert \) P_m\(\cos \phi\),
  \end{equation*}
where $j_m$ and $P_m$ stand for the spherical Bessel function and
Legendre polynomial of order $m$, respectively, and $\phi$ denotes the angle
between $x$ and $\xi$, i.e. $x\cdot \xi = \lvert x\rvert \lvert \xi\rvert
\cos \phi$ (see e.g. \cite{Handbook}). Denoting by $\om$ and $\om'$
the corresponding spherical  
coordinates of $x$ and $\xi$, respectively, we shall further expand $P_m$ via 
\begin{equation*}
  P_m(\cos \phi) = \frac{4\pi}{2m+1} \sum_{\ell =-m}^m Y^*_{\ell m}
  \(\om\) Y_{\ell m}(\om'),
\end{equation*}
where $Y_{\ell m}$ are the spherical harmonics. Note that, since
\begin{equation*}
  P_2(x)=\frac{1}{2}\(3x^2 -1\),
\end{equation*}
we have the following identity
\begin{equation*}
  K(x)= -2\, \frac{P_2(\cos \theta)}{|x|^3}=
  -2\, \sqrt{\frac{4\pi}{5}}\frac{Y_{20}(\om)}{|x|^3},  
\end{equation*}
where we recall that $\theta$ is the angle between $x$ and the dipole
axis $n$.  

Putting all of this together we can compute the Fourier
transform of $K$ via  
\begin{align*}
  \widehat K(\xi)\ = &\int_{\R^3}\sum_{m=0}^\infty
    (-i)^m (2m+1)j_m 
    \(\lvert x\rvert \lvert \xi\rvert \)P_m\(\cos \phi\)K(x)\, \D x\\
=& -2\sqrt{\frac{4\pi}{5}}\sum_{m=0}^\infty (-i)^m
    (2m+1)\int_{0}^\infty j_m(r|\xi|)\frac{1}{r^3} r^2
    \D r\times\\
& \int_{\om\in \Sph^2} P_m\(\cos \phi\)Y_{20}(\om)\, \D S(\om)\\
=& -2\sqrt{\frac{4\pi}{5}}\sum_{m=0}^\infty (-i)^m
    (2m+1)\int_{0}^\infty j_m(r|\xi|)\, \frac{\D r}{r} \times\\
&\frac{4\pi}{2m+1}\int_{\om\in \Sph^2}  \sum_{\ell =-m}^m Y^*_{\ell m}
  \(\om\) Y_{\ell m}(\om') Y_{20}(\om)\, \D S(\om)\\
=&-2\sqrt{\frac{4\pi}{5}} (-i)^2\times
    5\int_{0}^\infty j_2(r|\xi|)\, \frac{\D r}{r} \times
    \frac{4\pi}{5}Y_{20}(\om') , 
\end{align*}
where in the last equality we have used the fact that the $\{Y_{\ell
  m}\}$ form an orthonormal basis of $L^2(\Sph^2)$. We thus have, for
  $\xi\not =0$,  
\begin{align*}
 \widehat K(\xi)& =8\pi \sqrt{\frac{4\pi}{5}} Y_{20}(\om')
 \int_{0}^\infty j_2(r|\xi|)\frac{\D r}{r} 
=8\pi P_2\(\cos\Theta\) \int_{0}^\infty j_2(r)\frac{\D r}{r}  \\
&= 4\pi  \(3\cos^2\Theta
 -1\)\int_{0}^\infty j_2(r)\frac{\D r}{r}. 
\end{align*}
It remains to compute the radial integral in this formula. In view of
the identity 
\begin{equation*}
  \frac{\D}{\D r}\(\frac{j_1(r)}{r}\) = -\frac{j_{2}(r)}{r},
\end{equation*}
we have
\begin{equation*}
 \int_{0}^\infty j_2(r)\frac{\D r}{r} =\lim_{R\to 0}\int_{R}^\infty
 j_2(r)\frac{\D r}{r}= \lim_{R\to 0} \frac{j_1(R)}{R}.  
\end{equation*}
Recalling that 
\begin{equation*}
  j_1(R)= \frac{\sin R}{R^2}-\frac{\cos R}{R}, 
\end{equation*}
we thus find $\int_0^\infty j_2(r)/r \,  \D r = 1/3$, which finishes
the proof.  
\end{proof}

\section{Local in time existence}
\label{sec:local}

In view of the forthcoming dimension reduction analysis in Section
\ref{sec:reduction}, we allow the spatial dimension  
to be smaller than three and consider, instead of \eqref{eq:finalNLS},
the following initial value problem 
\begin{equation}\label{eq:NLSgen}
  \left \{
  \begin{split}
  i\d_t \psi + \frac{1}{2}\Delta \psi = & \, V_d(x)\psi +\l_1 \lvert
  \psi\rvert^2 \psi + \l_2\(K_d\ast \lvert\psi\rvert^2\)\psi,\quad
  x\in \R^d,\\ 
  \psi \big |_{t = 0} = & \, \varphi(x),
  \end{split}
  \right.
\end{equation}
where $1\le d\le 3 $.  We assume $\l_1,\l_2 \in \R$ and the potential
$V_d$ to be quadratic in $d$ dimensions, i.e. 
\begin{equation}\label{eq:potgen}
  V_d(x)= \frac{1}{2}\sum_{j=1}^d \om_j^2x_j^2, \quad \omega_j \in \R.
\end{equation}
The interaction kernel $K_d$ will not be specified in detail (except
that $K_3 \equiv K$). Rather, we shall only assume that  
the operator ${\mathcal K}_d:  u\mapsto K_d\ast u$ is bounded on
$L^2(\R^d)$, a property already known  
to be satisfied in $d=3$, see Lemma~\ref{lem:CZ}.  In $d=1,2$ we shall 
check that this holds true after we succeeded in 
deriving the precise expressions of $K_d$ via dimension reduction.

The following two important physical quantities are formally conserved
by the time-evolution associated to \eqref{eq:NLSgen}:  
\begin{align}
  \text{Mass: } M=&\lVert \psi(t)\rVert_{L^2}^2. \label{eq:mass} \\
\text{Energy: } E=&\frac{1}{2}\lVert\nabla \psi(t)\rVert_{L^2}^2
+\int_{\R^d}V_d(x) \lvert \psi(t,x)\rvert^2\D x + \frac{\l_1}{2}\lVert
\psi(t)\rVert_{L^4}^4  \label{eq:energy} \\
&+ \frac{\l_2}{2} \int_{\R^d} \(K_d\ast
\lvert\psi\rvert^2\)(t,x) \lvert \psi(t,x)\rvert^2\D x. \nonumber
\end{align}
Note that for the conservation of energy
(which can be derived formally by multiplying \eqref{eq:NLSgen} by
$\d_t \overline \psi$, taking real parts and integrating in $x$),
we use the fact that $K$ is even.
These two quantities naturally lead to the introduction of an energy
space associated with the linear case $\l_1=\l_2=0$:
\begin{equation*}
  \Sigma = \left\{ u \in L^2(\R^d)\ ;\quad \lVert u \rVert_\Sigma^2
  := \lVert u \rVert_{L^2}^2+ \lVert \nabla u \rVert_{L^2}^2 + \lVert
  x u \rVert_{L^2}^2 <\infty\right\}.
\end{equation*}
Let us ignore the dipole nonlinearity for the moment:
$\l_2=0$. $\Sigma$ is a natural space to study \eqref{eq:NLSgen},
where several results concerning the Cauchy problem are available
(see e.g. \cite{CazCourant}). Moreover, it is known that working only
in $L^2(\R^3)$ is  
not sufficient to prove (local) well-posedness; see e.g. \cite{CaInstab}
and references therein. 
Also note that the Cauchy problem is
ill-posed in $H^1(\R^d)$, due to the rotation of phase space induced
by the harmonic oscillator (see \cite{CaCauchy} for both linear and nonlinear
cases).  However, including the dipole-nonlinearity in this setting is not
completely straightforward as we can not  
allow for any (weak) derivative to be put on $K_d$, since the
corresponding convolution operator is  
no longer well defined. The only possible way to circumvent this
problem seems to be the use of  
Strichartz estimates. To this end, we shall introduce in the following
subsection the main technical tools  
needed later on in the local existence proof.

\subsection{Technical preliminaries}\label{sec:prelim}

Let $1\le d\le 3$, and $V_d$ be quadratic as in
\eqref{eq:potgen}. Introduce the group 
\begin{equation*}
  U(t)=e^{-itH},\quad \text{where }H= -\frac{1}{2}\Delta +V_d,
\end{equation*}
which generates the time-evolution for the linear problem.
We first remark that in view of Mehler's formula (see
e.g. \cite{Feyn}), the group $U(\cdot)$ is not only bounded on
$L^2(\R^d)$, but also enjoys dispersive properties for small
time. More precisely it holds 
\begin{equation}
\label{eq:dispest}
  \lVert U(t)\varphi \rVert_{L^\infty(\R^d)}\le \frac{C}{|t|^{d/2}}\lVert
  \varphi \rVert_{L^1(\R^d)},\quad \text{for }|t|\le \delta, 
\end{equation}
for some $\delta>0$. Note that for an harmonic potential as in
\eqref{eq:potgen}, $\delta$ is 
necessarily finite, since $H$ has eigenvalues. 

\begin{remark}
This dispersive
estimate \eqref{eq:dispest}  turns out to be valid for external
potentials which are not 
exactly quadratic (nor even positive): it suffices to have $V_d\in
C^\infty(\R^d;\R)$ with 
$\d^\alpha V_d\in L^\infty(\R^d)$ as soon as $|\alpha|\ge 2$ for the
above estimate to remain valid. In that case, the proofs we present
below remain valid, see e.g. \cite{CaCCM} and references
therein. 
\end{remark}

An important consequence of the dispersive estimate
\eqref{eq:dispest}, and of the fact that $U(\cdot)$ is  
unitary on $L^2$, is the existence of Strichartz
estimates, which require the following definition of admissible index pairs.
\begin{definition}\label{def:adm}
 A pair $(q,r)$ is {\bf admissible} if $2\le r
  \le\frac{2d}{d-2}$ (resp. $2\le r\le \infty$ if $d=1$, $2\le r<
  \infty$ if $d=2$)  
  and 
$$\frac{2}{q}=\delta(r):= d\left( \frac{1}{2}-\frac{1}{r}\right).$$
\end{definition}

Following \cite{GV85,Yajima87,KT}, we have:

\begin{lemma}[Strichartz estimates]\label{lem:strichartz} 
Let $(q,r)$, $(q_1,r_1)$ and~$ (q_2,r_2)$ be admissible pairs. Let $I$
be some finite time interval containing the origin, $0\in I$. \\
$1.$ There exists $C_r=C(r,I)$, such that for any $\varphi \in
L^2(\R^d)$ it holds 
\begin{equation}\label{eq:strich}
    \left\| U(\cdot)\varphi \right\|_{L^q(I;L^r)}\le C_r \|\varphi
    \|_{L^2}.
  \end{equation}
$2.$ There exists $C_{r_1,r_2}=C(r_1,r_2,I)$, such that for any $F\in
L^{q'_2}(I;L^{r'_2})$ it holds
\begin{equation}\label{eq:strichnl}
      \left\| \int_{I\cap\{s\leq
      t\}} U(t-s)F(s)\D s 
      \right\|_{L^{q_1}(I;L^{r_1})}\le C_{r_1,r_2} \left\|
      F\right\|_{L^{q'_2}(I;L^{r'_2})} .
    \end{equation}
\end{lemma}
As underscored above, taking for $\varphi$ an eigenfunction of the
(anisotropic)  
harmonic oscillator  shows that in general the constants $C_r$ and
$C_{r_1,r_2}$  
do depend on the length of the time interval $I$, unless all the
$\om_j$'s in \eqref{eq:potgen} are zero. 
\smallbreak
In the following we shall use two special admissible index pairs to
apply the above lemma, namely   
\begin{equation}\label{eq:pairs}
  (q,r)=(\infty,2)\quad \text{and}\quad (q,r)=\(\frac{8}{d},4\). 
\end{equation}
The idea for the proofs of local in time existence and uniqueness is
then to apply a fixed point argument on the 
Duhamel's formula associated to \eqref{eq:NLSgen}. This step turns out
to follow exactly the same lines as in the proof of local existence
without dipole, $\l_2=0$.  
To this end Strichartz estimates based on \eqref{eq:pairs} will be
invoked several times.

\subsection{Constructing a local solution in the energy space}

With the above technical tools in hand we can now state the first main
result of this work. 

\label{sec:constr}
\begin{proposition}\label{prop:localSigma}
  Let $1\le d\le 3$, $V_d$ be quadratic,  $\l_1,\l_2\in
  \R$, and $\varphi\in\Sigma$. Assume that the operator 
  ${\mathcal K}_d: 
  u\mapsto K_d\ast u$ is bounded on $L^2(\R^d)$. Then there exists
  $T$, depending only on (upper bounds for)
  $\| \varphi \|_{\Sigma}$, such that \eqref{eq:NLSgen} has a unique
  solution $\psi\in X_T$, where
  \begin{equation*}
    X_T =\left\{ \psi \in C([0,T];\Sigma)\ ;\ \psi,\nabla
    \psi, x\psi \in C([0,T];L^2(\R^d))\cap L^{8/d}([0,T];L^4(\R^d))
    \right\}. 
  \end{equation*}
Moreover, the mass $M$ and the energy $E$, defined in \eqref{eq:mass},
\eqref{eq:energy},  are  conserved for $t \in [0,T]$.
\end{proposition}
The space $L^{8/d}\([0,T];L^4\(\R^d\)\)$ is here to ensure
uniqueness. The appearance of this space will become
clear during the course of 
the proof. The main technical remark is that since ${\mathcal K}_d$ is
continuous on 
$L^2(\R^d)$ by assumption, the nonlocal term in \eqref{eq:NLSgen} can
be estimated like the local cubic nonlinearity.
\begin{proof} The proof of the conservations of mass and energy is omitted
here. We refer to \cite{CazCourant} for a general argument, easy to
adapt to the present case. Now, consider 
Duhamel's formula associated to \eqref{eq:NLSgen}  
\begin{equation*}
  \psi(t)=U(t) \varphi -i\l_1 \int_0^tU(t-s)\(|\psi|^2\psi\)(s)\D s
  -i\l_2\int_0^tU(t-s)\(\(K_d\ast |\psi|^2\)\psi\)(s)\D s , 
\end{equation*}
and denote $\Phi(\psi)(t)= U(t) \varphi +S_1(t)+S_2(t)$, with
\begin{align*}
  S_1(t)&= -i\l_1 \int_0^tU(t-s)\(|\psi|^2\psi\)(s)\D s,\\
  S_2(t)&= -i\l_2\int_0^tU(t-s)\(\(K_d\ast |\psi|^2\)\psi\)(s)\D s .
\end{align*}
Since at this stage, we are interested in a local result only, we may
assume that $T\in ]0,1]$. 
For $ \varphi\in \Sigma$, let $R=\| \varphi\|_{\Sigma}$, and introduce
\begin{align*}
  X_T(R)=\big\{ &\psi \in X_T\ ;\ \lVert A \psi
  \rVert_{L^\infty([0,T];L^2)}\le 4C_2R,\\
& \text{ and }  \lVert A \psi
  \rVert_{L^{8/d}([0,T];L^4)}\le 4C_4R,\quad \forall A\in \{{\rm Id},
  \nabla, x\}\big\}.  
\end{align*}
The constants $C_2$ and $C_4$ are those which appear in the first
point of Lemma~\ref{lem:strichartz}, with $I=[0,1]$. Apart from them we shall 
in the following denote all the irrelevant constants by $C$, whose value may
therefore change from one line to another. 

To prove
Proposition~\ref{prop:localSigma}, we first show that for $T\in ]0,1]$
sufficiently small, $\Phi$ leaves $X_T(R)$ invariant. Then, up to
demanding $T$ to be even smaller, we show that $\Phi$ is a contraction
on $X_T(R)$, for the topology of $L^{8/d}([0,T];L^4)$. As remarked in
\cite{Kato87}, $X_T(R)$, equipped 
with this topology, is complete, and thus
Proposition~\ref{prop:localSigma} follows, since $T$ will only depend
on $d$ and $R$.

\smallbreak

\noindent \emph{Step 1 (stability):} Let $\psi\in X_T(R)$ and denote
$L^q_TL^r= L^q([0,T];L^r(\R^d))$. The Strichartz
estimates yield:
\begin{align*}
  \| \Phi(\psi)\|_{L^\infty_TL^2} &\le
  \|U(\cdot) \varphi\|_{L^\infty_TL^2} +\|S_1\|_{L^\infty_TL^2}+
\|S_2\|_{L^\infty_TL^2}\\
&\le
  C_2R+C_{2,4} \left\lVert
  \lvert\psi\rvert^2\psi\right\rVert_{L^{8/(8-d)}_T L^{4/3}} +
C_{2,4} \left\lVert
  \(K_d\ast \lvert\psi\rvert^2\)\psi\right\rVert_{L^{8/(8-d)}_T L^{4/3}}.
\end{align*}
Similarly, it holds
\begin{equation*}
  \| \Phi(\psi)\|_{L^{8/d}_TL^4} \le
  C_4R+C_{4,4} \left\lVert
  \lvert\psi\rvert^2\psi\right\rVert_{L^{8/(8-d)}_T L^{4/3}} +
C_{4,4} \left\lVert
  \(K_d\ast \lvert\psi\rvert^2\)\psi\right\rVert_{L^{8/(8-d)}_T L^{4/3}}.
\end{equation*}
We shall only work on the latter estimate as the $L^\infty_TL^2$ estimate can
handled exactly in the same way. Denote $(q,r)=(8/d,4)$ and remark the
following  identities 
\begin{equation}\label{eq:kident}
  \frac{3}{4}= \frac{1}{r'}=\frac{3}{r}\quad ;\quad
  \frac{8-d}{8}=\frac{1}{q'}=\frac{1}{q}+\frac{2}{k}, \text{ with }k=
  \frac{8}{4-d}. 
\end{equation}
H\"older's inequality thus yields
\begin{equation*}
  \left\lVert
  \lvert\psi\rvert^2\psi\right\rVert_{L^{8/(8-d)}_T L^{4/3}} \le
  \lVert \psi\rVert_{L^k_TL^r}^2\lVert \psi\rVert_{L^q_TL^r}.  
\end{equation*}
By assumption, the operator ${\mathcal K}_d: u\mapsto K_d\ast |u|^2$
is bounded on $L^2(\R^d)$ and so we have (recall that $r=4$):
\begin{align*}
  \left\lVert
  \(K_d\ast \lvert\psi\rvert^2\)\psi\right\rVert_{L^{8/(8-d)}_T
  L^{4/3}}&\le \left\lVert
  K_d\ast \lvert\psi\rvert^2\right\rVert_{L^{k/2}_T L^{r/2}}\lVert
  \psi\rVert_{L^q_TL^r}\\
&\le  C \left\lVert
  \lvert\psi\rvert^2\right\rVert_{L^{k/2}_T L^{r/2}}\lVert
  \psi\rVert_{L^q_TL^r} \le C\lVert \psi\rVert_{L^k_TL^r}^2\lVert
  \psi\rVert_{L^q_TL^r}. 
\end{align*}
We note that this is exactly the same estimate as for the local cubic
term, up to the norm of ${\mathcal K}_d$. 
From Sobolev embedding (in space), we get
\begin{equation*}
  \lVert \psi\rVert_{L^k_TL^r} \le C \lVert \psi\rVert_{L^k_TH^1}\le C
  T^{1/k} \lVert \psi\rVert_{L^\infty_TH^1}\le CT^{1/k} C_2 R.
\end{equation*}
Therefore,
\begin{equation*}
  \|\Phi(\psi)\|_{L^{q}_TL^r} \le C_4R + CT^{2/k}R^3,
\end{equation*}
and if $T\in ]0,1]$ is sufficiently small, $\|
\Phi(\psi)\|_{L^{q}_TL^r} \le 4C_4 R$. 

Next, to estimate $\nabla \Phi(\psi)$, we denote by $[A,B]=AB-BA$ the
usual commutator and compute  
\begin{equation*}
  [\nabla,H]= [\nabla,V]= \nabla V ,\quad [x,H]=\frac{1}{2}[x,\Delta]
  = -\nabla. 
\end{equation*}
By assumption, $|\nabla V(x)|\le C |x|$, which shows that we will
indeed get a closed family of 
estimates for $\nabla \Phi (\psi)$ and $x \Phi (\psi)$. For instance,
\begin{equation}
  \label{eq:gradient}
  \begin{aligned}
  \nabla \Phi(\psi)(t) =& \ U(t)\nabla \varphi -i\l_1
  \int_0^tU(t-s)\nabla \(|\psi|^2\psi\)(s)\D s \\
&  -i\l_2\int_0^tU(t-s)\nabla \(\(K_d\ast |\psi|^2\)\psi\)(s)\D s \\
& -i \int_0^t U(t-s)\(\Phi(\psi(s))\nabla V\)\D s.
\end{aligned}
\end{equation}
From the Strichartz estimates, we thus have
\begin{equation*}
  \|U(\cdot )\nabla \varphi \|_{L^\infty_TL^2}\le C_2R,
\end{equation*}
and for the local cubic nonlinearity we get
\begin{align*}
  \left\lVert \int_0^tU(t-s)\nabla \(|\psi|^2\psi\)(s)\D s
  \right\rVert_{L^\infty_TL^2} &\le C_{2,4}\left\lVert
  \nabla \(|\psi|^2\psi\)
  \right\rVert_{L^{q'}_TL^{r'}}\le C\left\lVert
  |\psi|^2\lvert \nabla \psi\rvert
  \right\rVert_{L^{q'}_TL^{r'}}\\
&\le C \lVert
  \psi\rVert_{L^k_TL^r}^2\lVert \nabla\psi\rVert_{L^q_TL^r}\le CR\lVert
  \psi\rVert_{L^k_TL^r}^2 \\
&\le CT^{2/k}R^3, 
\end{align*}
where we have used the same computations as above, on $\Phi(\psi)$. For
the dipole-nonlinearity, we similarly obtain 
\begin{align*}
  \Big\lVert \int_0^t&U(t-s)\nabla \(\(K_d\ast|\psi|^2\)\psi\)(s)\D s
  \Big\rVert_{L^\infty_TL^2} \le C_{2,4}\left\lVert
  \nabla \(\(K_d\ast|\psi|^2\)\psi\)
  \right\rVert_{L^{q'}_TL^{r'}}\\
&\le C \left\lVert
  \(K_d\ast|\psi|^2\)\nabla \psi
  \right\rVert_{L^{q'}_TL^{r'}}+ 
C \left\lVert
  \(K_d\ast\nabla |\psi|^2\)\psi
  \right\rVert_{L^{q'}_TL^{r'}}\\
&\le C \lVert K_d\ast |\psi|^2
  \rVert_{L^{k/2}_TL^{r/2}}\lVert \nabla \psi\rVert_{L^q_TL^r}
+ C \lVert K_d\ast\nabla |\psi|^2
 \rVert_{L^{(1/q+1/k)^{-1}}_TL^{r/2}}\lVert\psi\rVert_{L^k_TL^r}\\
&\le C \lVert |\psi|^2
  \rVert_{L^{k/2}_TL^{r/2}}\lVert \nabla \psi\rVert_{L^q_TL^r}
+ C \lVert \nabla |\psi|^2
 \rVert_{L^{(1/q+1/k)^{-1}}_TL^{r/2}}\lVert\psi\rVert_{L^k_TL^r}\\
&\le C \lVert \psi
  \rVert_{L^{k}_TL^{r}}^2\lVert \nabla \psi\rVert_{L^q_TL^r}
+ C \lVert \psi \nabla \psi
 \rVert_{L^{(1/q+1/k)^{-1}}_TL^{r/2}}\lVert\psi\rVert_{L^k_TL^r}\\
&\le C \lVert \psi
  \rVert_{L^{k}_TL^{r}}^2\lVert \nabla \psi\rVert_{L^q_TL^r}\le CR \lVert \psi
  \rVert_{L^{k}_TL^{r}}^2\le CT^{2/k}R^3.
\end{align*}
This is the same type of estimate as for the local cubic
term. Finally, the last term in \eqref{eq:gradient} is estimated by
\begin{align*}
 \Big\lVert \int_0^t U(t-s)\(\Phi(\psi(s))\nabla V\)\D s
 \Big\rVert_{L^\infty_TL^2} \le C_{2,2}\lVert \Phi(\psi)\nabla
 V\rVert_{L^1_TL^2} \le C T \lVert x
 \Phi(\psi)\rVert_{L^\infty_TL^2}. 
\end{align*}
All in all, we find:
\begin{equation*}
  \lVert\nabla \Phi(\psi)\rVert_{L^\infty_TL^2} \le C_2 R + C T^{2/k}
  R^3 + CT \lVert x
 \Phi(\psi)\rVert_{L^\infty_TL^2}.
\end{equation*}
Analogously we obtain (with slightly shorter computations)
\begin{equation*}
  \lVert x\Phi(\psi)\rVert_{L^\infty_TL^2} \le C_2 R + C T^{2/k}
  R^3 + CT \lVert \nabla
 \Phi(\psi)\rVert_{L^\infty_TL^2}.
\end{equation*}
Summing up these two inequalities and taking $T\in ]0,1]$ sufficiently
small, we arrive at
\begin{equation*}
  \lVert\nabla \Phi(\psi)\rVert_{L^\infty_TL^2} + \lVert
  x\Phi(\psi)\rVert_{L^\infty_TL^2} \le 4C_2 R, 
\end{equation*}
and using the Strichartz estimates again, we also have
\begin{align*}
  \lVert\nabla \Phi(\psi)\rVert_{L^q_TL^r} &\le C_4 R + C T^{2/k}
  R^3 + CT \lVert x
 \Phi(\psi)\rVert_{L^\infty_TL^2},\\
\lVert x\Phi(\psi)\rVert_{L^q_TL^r} &\le C_4 R + C T^{2/k}
  R^3 + CT \lVert \nabla
 \Phi(\psi)\rVert_{L^\infty_TL^2}.
\end{align*}
Up to taking $T$ even smaller, we therefore see that $\Phi$ leaves $X_T(R)$
stable. 
\smallbreak

\noindent \emph{Step 2 (contraction):} For $\psi_1$ and $\psi_2$ in
$X_T(R)$, we have, from Strichartz estimates and the above
computations, 
\begin{align*}
   \lVert    \Phi & (\psi_1)  -\Phi(\psi_2)\rVert_{L^q_TL^r} \\
   \le & \,  C_{4,4} \lVert
  |\psi_1|^2\psi_1 -|\psi_2|^2\psi_2\rVert_{L^{q'}_TL^{r'}} +  C_{4,4} \lVert 
  \(K_d\ast |\psi_1|^2\)\psi_1 -\(K_d\ast
  |\psi_2|^2\)\psi_2\rVert_{L^{q'}_TL^{r'}} \\
 \le  & \, C \left\lVert
  \(|\psi_1|^2+|\psi_2|^2\)\lvert \psi_1-\psi_2\rvert
\right\rVert_{L^{q'}_TL^{r'}} + C \left\lVert \(K_d\ast\( 
  |\psi_1|^2-|\psi_2|^2\)\)\psi_1
\right\rVert_{L^{q'}_TL^{r'}} \\
&+C \left\lVert
  \(K_d\ast|\psi_2|^2\)\(\psi_1-\psi_2\)
\right\rVert_{L^{q'}_TL^{r'}}\\
\le & \, C T^{2/k}R^2 \lVert \psi_1-\psi_2\rVert_{L^q_TL^r} + C \left\lVert 
  |\psi_1|^2-|\psi_2|^2\right\rVert_{L^{(1/q+1/k)^{-1}}_TL^{r}}\left\lVert
  \psi_1
\right\rVert_{L^{k}_TL^{r}}\\
&+ C T^{2/k}R^2 \lVert \psi_1-\psi_2\rVert_{L^q_TL^r} \\
\le & \, C T^{2/k}R^2 \lVert \psi_1-\psi_2\rVert_{L^q_TL^r}.
\end{align*}
Therefore, $\Phi$ is a contraction provided that $T$ is
sufficiently small, and hence the result follows. 
\end{proof}

\section{Global in time existence results}
\label{sec:global}

Having established the existence and uniqueness of solutions to
\eqref{eq:NLSgen} locally in time, we now turn our attention to  
global in time results. To this end we first note that the only
obstruction for global existence in $\Sigma$ is 
the possible unboundedness of $\nabla \psi$ in $L^2(\R^d)$. Indeed, in
$\R^d$ with $1\le d\le 3$, the Gagliardo--Nirenberg inequality (see
e.g. \cite{CazCourant}) yields
\begin{equation*}
 \lVert u \rVert_{L^4}^4 \le C  \lVert u \rVert_{L^2}^{4-d} \, \lVert
 \nabla u
 \rVert_{L^2}^d. 
\end{equation*}
Thus, the boundedness of ${\mathcal K}_d$ on $L^2(\R^d)$ together with the
Cauchy--Schwarz inequality implies
\begin{equation}
\label{eq:dipolest}
\begin{split}
 \left\lvert \int_{\R^d} \(K_d\ast
\lvert u \rvert^2\)(x) \lvert u (x)\rvert^2\D x \right\rvert &\le
\left\lVert K_d\ast
\lvert u \rvert^2\right\rVert_{L^2} \left\lVert
\lvert u \rvert^2\right\rVert_{L^2} \le C \lVert u \rVert_{L^4}^4 \\
&\le C
\lVert u \rVert_{L^2}^{4-d} \, \lVert 
 \nabla u
 \rVert_{L^2}^d,
 \end{split}
\end{equation}
where we have used Gagliardo--Nirenberg inequalities. So we see that
in view of the conservation of the mass, the (conserved) energy is
the sum of four terms, three of which are bounded provided that
$\nabla \psi$ remains bounded in $L^2(\R^d)$. Therefore, the fourth
term is bounded, that is, $x\psi(t,x)$ is bounded in
$L^2(\R^d)$. This shows that unless $\lVert \nabla
\psi(t)\rVert_{L^2}$ becomes unbounded, $\|\psi(t)\|_{\Sigma}$ is a
continuous function of $t\ge 0$. This directly yields the following corollary:
\begin{corollary}\label{cor:maximal}
  Under the same assumptions as in Proposition \ref{prop:localSigma},
  there exists a $T_*\in \R_+\cup \{\infty\}$, such 
  that \eqref{eq:NLSgen} 
  has a unique maximal solution in
\begin{equation*}
   \left\{ \psi \in C([0,T_*[;\Sigma)\ ;\ \psi,\nabla
    \psi, x\psi \in C([0,T_*[;L^2(\R^d))\cap L^{8/d}_{\rm
    loc}([0,T_*[;L^4(\R^d)) \right\},  
  \end{equation*}
such that $\psi_{\mid t=0}=\varphi$. It is maximal in the sense that
if $T_*<\infty$, then  
\begin{equation*}
  \left\lVert \nabla \psi(t)\right\rVert_{L^2}\Tend t {T_*}+\infty. 
\end{equation*}
\end{corollary}
Having in mind this result, we are now able to study the
global in time existence of solutions to \eqref{eq:NLSgen}  
depending on the spatial dimension. 
In the following we shall first treat the case $d=1$, where  
the picture is much more concise,  before moving on to the case $d=3$
(the case $d=2$, is similar to  
the one of three spatial dimensions and thus omitted for
simplicity). Note that, in view of the  
dimensional reduction discussed in Section \ref{sec:reduction}, the
one-dimensional case is not purely academic.

\subsection{Global existence for $d=1$}\label{sec:global1d}

When $d=1$, the solution constructed in the previous section indeed
exists for all time. More precisely we have:

\begin{corollary}\label{cor:global1DSigma}
  Suppose that in Proposition~\ref{prop:localSigma}, $d=1$. Then the
  solution to 
  \eqref{eq:NLSgen} is global in time, i.e. $T_*=\infty$ in
  Corollary~\ref{cor:maximal}.  
\end{corollary}
\begin{proof}
  This result stems from Gagliardo--Nirenberg inequality, like its
  classical counterpart when $\l_2=0$. In view of
  Corollary~\ref{cor:maximal}, we just need an \emph{a priori}
  estimate for $\|\nabla \psi(t)\|_{L^2}$. Indeed the conservation of energy
  yields
  \begin{align*}
    \lVert \nabla \psi(t)\rVert_{L^2}^2 & \le 2E + |\l_1|\lVert
    \psi(t)\rVert_{L^4}^4 + |\l_2| \left\lvert \int_{\R} \(K_1\ast
\lvert\psi\rvert^2\)(t,x) \lvert \psi(t,x)\rvert^2\D x\right\rvert\\
&\le 2E + C  \lVert \psi(t)\rVert_{L^2}^3 \lVert \nabla
    \psi(t)\rVert_{L^2}, 
  \end{align*}
where we have again used the estimates \eqref{eq:dipolest}. Since  $M=
\lVert \psi(t)\rVert_{L^2}^2$ is   
constant in time, this inequality directly shows that $\|\nabla
\psi(t)\|_{L^2}$ remains bounded and hence the  result is proven. 
\end{proof}
When  $d=1$, a more general setting is possible, which allows us
to work with initial data that are only  
in $L^2$ but not necessarily in $\Sigma$ (corresponding to solutions
with possibly infinite energy). The following result is an adaptation
of the main result in \cite{TsutsumiL2}.

\begin{theorem}\label{thm:globalL2}
  Let $d=1$, $V_1$ be quadratic,  and $\l_1,\l_2\in
  \R$. Assume that the operator  ${\mathcal K}_1: 
  u\mapsto K_1\ast u$ is bounded on $L^2(\R)$. Then, for any 
  $ \varphi \in L^2(\R)$, \eqref{eq:NLSgen}
  has a unique solution 
  \begin{equation*}
    \psi\in C\(\R_+;L^2\)\cap L^{8}_{\rm loc}\(\R_+;L^4\(\R\)\).
  \end{equation*}
Moreover its total mass $M$ is
 independent of $t\ge 0$.
\end{theorem}
\begin{proof}
 We shall first prove that there exists a $T$, depending only on 
  $\| \varphi\|_{L^2(\R)}$, such that \eqref{eq:NLSgen} has a unique
  solution 
  \begin{equation*}
    \psi\in C\([0,T];L^2\)\cap L^{8}\([0,T];L^4\(\R\)\), 
  \end{equation*} 
To this end we resume the same scheme as in the proof of
Proposition~\ref{prop:localSigma} 
above, but now $R=\| \varphi\|_{L^2}$ and 
\begin{align*}
  Y_T(R)=\big\{ &\psi \in L^\infty([0,T];L^2)\cap L^8([0,T];L^4)\ ;\
  \lVert \psi 
  \rVert_{L^\infty([0,T];L^2)}\le 2C_2R,\\
& \text{ and }  \lVert  \psi
  \rVert_{L^{8}([0,T];L^4)}\le 2C_4R\big\}.  
\end{align*}
We simply notice that for $d=1$, the identities \eqref{eq:kident}
yield $k=8/3$, and so H\"older's 
inequality in time implies
\begin{equation*}
  \lVert \psi\rVert_{L^k_TL^r} = \lVert \psi\rVert_{L^{8/3}_TL^4}\le
  T^{1/4}  \lVert \psi\rVert_{L^{8}_TL^4}=T^{1/4}\lVert
  \psi\rVert_{L^{q}_TL^r}. 
\end{equation*}
Following the same lines as in the previous paragraph, we see that
$\Phi$ leaves $Y_T(R)$ stable, provided that $T$ is sufficiently
small. The proof of contraction is the same as in the previous
paragraph, up to the modification of the estimate for $\lVert
\psi\rVert_{L^k_TL^r}$. This yields a local in time existence and
uniqueness result and since the existence time 
depends only  on $\| \varphi\|_{L^2(\R)}$, the conservation of the total mass
$M$ directly implies the global in time existence.
\end{proof}
In summary the Gross--Pitaevskii equation \eqref{eq:NLSgen} is globally
well-posed for $d=1$. The  
three dimensional case is more involved, though, as we shall see.

\subsection{Global existence for $d=3$ (stable regime)}\label{sec:global3d}

We can now turn to the physically most important case $d= 3$. 
\begin{theorem}\label{thm:global3DSigma-1}
  Under the same assumptions as in Proposition~\ref{prop:localSigma},
  suppose that in 
  addition $\l_1\ge \frac{4\pi}{3} \l_2\ge 0$. Then the 
  solution is global in time, i.e. $T_*=\infty$.  
\end{theorem}
In the following, the situation where $\l_1\ge \frac{4\pi}{3} \l_2\ge
0$ will be called the stable regime (note that $\lambda_2 >0$  
corresponds to the actual physical situation).
\begin{proof}  We first note that from Plancherel's formula for $\rho
  =\lvert\psi\rvert^2$ we get 
\begin{equation*}
  \|\psi(t)\|_{L^4}^4 =\|\rho(t)\|_{L^2}^2 = \frac{1}{(2\pi)^3}
  \left\lVert\widehat \rho(t)\right\rVert_{L^2}^2 .
\end{equation*}
  Then we simply use the conservation of the energy $E$, as defined in
  \eqref{eq:energy}, to estimate 
  \begin{align*}
    \lVert \nabla \psi(t)\rVert_{L^2}^2 & =2E -\int_{\R^3} V_3(x)\lvert
    \psi(t,x)\rvert^2\D x - \l_1\lVert
    \psi(t)\rVert_{L^4}^4 \\
&\phantom{=}- \l_2 \int_{\R^3} \(K_1\ast
\lvert\psi\rvert^2\)(t,x) \lvert \psi(t,x)\rvert^2\D x\\
&\le 2E -\frac{1}{(2\pi)^3} \int_{\R^3} \(\l_1 +\l_2 \widehat K(\xi)\)\lvert
    \widehat \rho(\xi)\rvert^2 \D \xi, 
  \end{align*}
invoking the above given identity. 
Recalling the explicit formula for $\widehat K$ computed in
Lemma~\ref{lem:FourierK}, we obtain 
\begin{equation*}
  \int_{\R^3} \(\l_1 +\l_2 \widehat K(\xi)\)\lvert
    \widehat \rho(\xi)\rvert^2 \D \xi\ge \int_{\R^3} \(\l_1
    -\frac{4\pi}{3} \l_2 \)\lvert  
    \widehat \rho(\xi)\rvert^2 \D \xi.
\end{equation*}
By assumption, this quantity is non-negative and hence the
\emph{a priori} estimate
\begin{equation*}
  \lVert \nabla \psi(t)\rVert_{L^2}^2 \le 2E
\end{equation*}
holds true. 
The result then follows directly from Corollary~\ref{cor:maximal}. 
\end{proof}

Having established this result, it is natural to ask what happens in
the unstable regime $\l_1 < \frac{4\pi}{3} \l_2$.  
As will shall see, in general we cannot expect a global in time
result there, since finite time blow-up of solutions (in the sense of  
Corollary~\ref{cor:maximal}) may occur. In particular, it is clear
that if blow-up occurs, 
then the problem of dimension reduction ceases to make sense. The 
next section is devoted to the study of this problem.

\section{The unstable regime}\label{sec:unstable}

When $d=3$ (or $2$), the solution constructed in
Proposition~\ref{prop:localSigma} need not remain in $\Sigma$ for all
time. This will be seen from using a general virial computation (see
\cite{CazCourant}) and following the approach by Zakharov \cite{Z} and
Glassey \cite{Glassey}. 

\subsection{Finite time blow-up}\label{sec:virial}

As a preliminary, we check that the energy $E$ may be negative. 
\begin{lemma}\label{lem:negative}
  Let $d=3$. Assume that 
  \begin{equation*}
    \l_1 < \cst \l_2.
  \end{equation*}
There exists $\varphi\in \Sigma$ such that $E<0$, where the
energy $E$ is as in Proposition~\ref{prop:localSigma}. 
\end{lemma}
\begin{proof}
  We first rewrite the last term in the energy, thanks to Plancherel
  formula:
  \begin{equation*}
    \int_{\R^3} \(K\ast
\lvert\varphi\rvert^2\)(x) \lvert \varphi(x)\rvert^2dx =
\frac{1}{(2\pi)^3}\int_{\R^3} \widehat K(\xi) \lvert\widehat
\rho(\xi)\rvert^2d\xi,  
  \end{equation*}
where we have denoted $\rho=\lvert \varphi\rvert^2$. In view of
Lemma~\ref{lem:FourierK}, we infer
\begin{equation*}
  \int_{\R^3} \(K\ast
\lvert\varphi\rvert^2\)(x) \lvert \varphi(x)\rvert^2dx =
\frac{1}{(2\pi)^2} \int_{\R^3}\(3\frac{\xi_3^2}{\lvert
  \xi\rvert^2}-1\) \lvert\widehat
\rho(\xi)\rvert^2d\xi.
\end{equation*}
The idea of the proof then consists in choosing $\varphi$ so that
$\widehat \rho$ has little mass on $\{\xi_3=0\}= \{\xi\cdot
n=0\}$. Indeed, using Plancherel formula again, we can rewrite the
energy as
\begin{equation}
\label{eq:encomp}
\begin{aligned}
  E=& \frac{1}{2}\lVert\nabla \varphi \rVert_{L^2}^2
+\int_{\R^3}V(x) \lvert  \varphi(x)\rvert^2\D x \\
&+ \frac{1}{2(2\pi)^3}  \int_{\R^3}\(\l_1 +\frac{4\pi}{3} \l_2
\(3\frac{\xi_3^2}{\lvert 
  \xi\rvert^2}-1\) \)\lvert\widehat
\rho(\xi)\rvert^2\D \xi. 
\end{aligned}
\end{equation}
Introduce a parameter $\eps>0$, and fix some functions $f \in
\Sch(\R^2)$ and $g\in \Sch(\R)$ independent of $\eps$. Set
\begin{equation*}
  \varphi(x) =\eps^{\alpha/2} f(x_1,x_2)g(\eps x_3),
\end{equation*}
for some constant $\alpha$ to be fixed later. We have
\begin{equation*}
  \rho(x) = \lvert \varphi(x)\rvert^2 = \eps^\alpha \lvert
  f(x_1,x_2)\rvert^2 \lvert g(\eps x_3)\rvert^2 \ ,\quad \widehat
  \rho(\xi) = \eps^{\alpha-1} F(\xi_1,\xi_2)G\(\frac{\xi_3}{\eps}\),
\end{equation*}
where $F$ and $G$ denote the Fourier transforms of $|f|^2$ and
$|g|^2$, in $\Sch(\R^2)$ and $\Sch(\R)$, respectively. We now measure
the order of magnitude, as $\eps\to 0$, of each term in the energy:

\noindent $\bullet$ Kinetic energy: the leading order term corresponds to the
differentiation with respect to $x_1$ or $x_2$.
\begin{equation*}
 \lVert\nabla \varphi \rVert_{L^2}^2 \approx  \eps^\alpha
 \int_{\R^3}\lvert\nabla f(x_1,x_2)\rvert^2 \lvert g(\eps x_3)\rvert^2
 dx\approx \eps^{\alpha-1}. 
\end{equation*}
\noindent $\bullet$ Potential energy: the leading order term corresponds
 to the $x_3$ component. 
\begin{equation*}
 \int_{\R^3}V(x) \lvert \varphi(x)\rvert^2dx\approx  \eps^\alpha
 \int_{\R^3}x_3^2\lvert f(x_1,x_2)\rvert^2 \lvert g(\eps x_3)\rvert^2
 dx\approx \eps^{\alpha-3}. 
\end{equation*}
\noindent $\bullet$ Cubic nonlinear term: 
\begin{equation*}
  \int \lvert \widehat \rho(\xi)\rvert^2d\xi =\eps^{2\alpha-2}\int
  \lvert F(\xi_1,\xi_2)\rvert^2 \left \lvert G\(\frac{\xi_3}{\eps}\)
  \right\rvert^2  d\xi \approx \eps^{2\alpha-1}.
\end{equation*}
\noindent $\bullet$ Finally we compute: 
\begin{align*}
 \int \frac{\xi_3^2}{\lvert 
  \xi\rvert^2}\lvert \widehat \rho(\xi)\rvert^2d\xi &= \eps^{2\alpha-2}
  \int \frac{\xi_3^2}{\xi_1^2+\xi_2^2+\xi_3^2}\lvert
  F(\xi_1,\xi_2)\rvert^2 \left \lvert G\(\frac{\xi_3}{\eps}\) 
  \right\rvert^2  d\xi\\
&=\eps^{2\alpha-1}
  \int \frac{\eps^2\xi_3^2}{\xi_1^2+\xi_2^2+\eps^2\xi_3^2}\lvert
  F(\xi_1,\xi_2)\rvert^2 \left \lvert G\(\xi_3\) 
  \right\rvert^2  d\xi=o\(\eps^{2\alpha-1}\),
\end{align*}
where we have used Lebesgue's Dominated Convergence Theorem. We therefore have:
\begin{equation*}
  E\approx \eps^{\alpha-1} + \eps^{\alpha-3} + \(\l_1 -\cst
  \l_2\)\eps^{2\alpha-1} + o\(\eps^{2\alpha-1}\).
\end{equation*}
If we choose $\alpha<-2$, then the leading order term is the third
one, and the lemma follows. 
\end{proof}
As a consequence we are now able to prove finite time blow-up for a
certain class of initial data. 
\begin{theorem}\label{thm:virial}
  Let $d=3$ and $ \varphi\in \Sigma$. Denote $\underline\om=\min
  \om_j$ and assume that 
  \begin{equation*}
 3  E\le \underline \om^2 \lVert x  \varphi\rVert_{L^2}^2. 
  \end{equation*}
Then the solution $\psi$ blows up in 
finite time, i.e. $T_*<\infty$  in
Corollary~\ref{cor:maximal}. More precisely, we can estimate $T_*\le
\pi/(2\underline\om)$.   
\end{theorem}
  From Lemma~\ref{lem:negative} we know that for $\l_1<\frac{4\pi}{3} \l_2$ 
  we can always choose initial data such that $E<0$ and thus enforce
  finite time blow-up.   
  In other words, blow-up may occur even in situations where we have a
  defocusing (local) cubic nonlinearity  
  ($\l_1>0$) and a positive coupling constant $\l_2>0$ (the 
  physical case), provided that $\l_1<\frac{4\pi}{3} \l_2$ holds true
  (and the initial energy is sufficiently small).  
  Recalling the well known fact that when
  $\l_1>0$ and $\l_2=0$, finite time blow-up cannot occur
  (i.e. $T_*=\infty$), this shows that the presence of the dipole-term
  may indeed cause collapse of the Bose--Einstein condensate. 

\begin{proof}
  The proof is based on the virial computation. Set 
  \begin{equation*}
    y(t) = \int_{\R^3} \lvert x\rvert^2 \lvert \psi(t,x)\rvert^2\D x. 
  \end{equation*}
Then, following \cite{CazCourant} (with slightly different notations), we have,
since $K$ is even:
\begin{align*}
  \ddot y(t)= &\ 4E + \l_1\|\psi(t)\|_{L^4}^4 -8\int_{\R^3} V(x)\lvert
  \psi(t,x)\rvert^2\D x \\
&-2\l_2\int_{\R^3} \(\(K+\frac{1}{2}x\cdot \nabla
  K\)\ast\lvert \psi\rvert^2\)(t,x)\lvert \psi(t,x)\rvert^2\D x.
\end{align*}
We again use Plancherel's formula with $\rho =\lvert\psi\rvert^2$ to write
\begin{align*}
\int\(\(K+\frac{1}{2}x\cdot \nabla
  K\)\ast \rho\)(t,x) \rho(t,x)\D x= \frac{1}{(2\pi)^3} \int \(\widehat
  K +\frac{1}{2}\widehat{x\cdot \nabla K}\) \lvert\widehat
  \rho(\xi)\rvert^2\D \xi
\end{align*}
and compute, for $d=3$, 
\begin{align*}
  \widehat{x\cdot \nabla K}(\xi) = -\DIV \(\xi \widehat K(\xi)\) =
  -3\widehat K(\xi) -\xi\cdot \nabla\widehat K(\xi). 
\end{align*}
From Lemma~\ref{lem:FourierK}, we infer $ \xi\cdot \nabla\widehat
K(\xi)=0$ and thus
\begin{equation*}
  \ddot y(t)=  4E + \frac{\l_1}{(2\pi)^3}\left\lVert\widehat
  \rho(t)\right\rVert_{L^2}^2 -8\int_{\R^3} V(x)\lvert 
  \psi(t,x)\rvert^2\D x +\frac{\l_2}{(2\pi)^3}\int_{\R^3} \widehat K(\xi)
  \lvert\widehat \rho(\xi)\rvert^2\D \xi.
\end{equation*}
This can be rewritten as
\begin{equation*}
  \ddot y(t)+4\underline \om^2 y(t) =  f(t),
\end{equation*}
with  $\underline\om=\min
  \om_j$ and
\begin{equation*}
f(t)=4E -4\sum_{j=1}^3
  \(\om_j^2-\underline \om^2\)\int_{\R^3} x_j^2\lvert \psi(t,x)\rvert^2\D x
 +\frac{1}{(2\pi)^3}\int_{\R^3}
  \(\l_1 + \l_2 \widehat K(\xi)\) \lvert\widehat
  \rho(\xi)\rvert^2\D \xi. 
\end{equation*}
Recalling the energy as written as in \eqref{eq:encomp}, we first note
that the source term $f(t)$ can be estimated via 
\begin{align*}
  f(t)&\le 4E +\frac{1}{(2\pi)^3}\int_{\R^3}
  \(\l_1 + \l_2 \widehat K(\xi)\) \lvert\widehat
  \rho(\xi)\rvert^2\D \xi \\
&\le 6E -  \lVert\nabla \psi(t)\rVert_{L^2}^2
-2\int_{\R^3}  V(x) \lvert \psi(t,x)\rvert^2\D x \le 6E. 
\end{align*}
On the other hand we have
\begin{equation*}
  y(t)=y(0) \cos\(2\underline \om t\) +\dot y(0)\frac{\sin
  \(2\underline \om t\)}{2\underline \om} +\int_0^t \frac{\sin
  \(2\underline \om (t-s)\)}{2\underline \om}f(s)\, \D s. 
\end{equation*}
Suppose now that $T_*>\pi/(2\underline \om)$. Then we can consider
$t=\pi/(2\underline \om)$ in the above relation. This yields, since
the sine function in the last integral remains non-negative on
$[0,t]$:
\begin{equation*}
  y\(\frac{\pi}{2\underline \om}\) \le -y(0) +6E
  \int_0^{\pi/(2\underline \om)} \frac{\sin 
  \(\pi -2\underline \om s\)}{2\underline
  \om}\, \D s=-y(0)+\frac{3E}{\underline \om^2}. 
\end{equation*}
We note that the left hand side must be positive, since
$\psi(t,\cdot)\in \Sigma$. By assumption, the right hand side is
non-positive. This yields a contradiction and hence the result. 
\end{proof}

\begin{remark}
To shed more light on the conditions for blow-up, 
consider the case of an isotropic trapping potential: $\om_j =\om$
  for $1\le j\le 3$. The condition of Theorem~\ref{thm:virial}
  then reads
  \begin{equation*}
\frac{3}{2}\lVert\nabla  \varphi \rVert_{L^2}^2
+\frac{\om^2}{2}\lVert x  \varphi \rVert_{L^2}^2 + \frac{3\l_1}{2}\lVert
 \varphi\rVert_{L^4}^4 
+ \frac{3\l_2}{2} \int_{\R^3} \(K_d\ast
\lvert \varphi\rvert^2\)(x) \lvert  \varphi(x)\rvert^2\D x \le 0.    
  \end{equation*}
In view of Plancherel's formula, this is equivalent to
\begin{equation*}
3\lVert\nabla  \varphi \rVert_{L^2}^2
+\om^2\lVert  x  \varphi \rVert_{L^2}^2 + \frac{3}{(2\pi)^3}\int_{\R^3}
\(\l_1 +\l_2 
\widehat K(\xi)\) \lvert \widehat \rho(\xi)\rvert^2\D \xi\le 0.    
  \end{equation*}
From Lemma~\ref{lem:FourierK}, when $\l_2>0$ (the physical case), this
is possible only if $\l_1<\frac{4\pi}{3} \l_2$. 
\end{remark} 

We already know that when  
$\l_1\ge \frac{4\pi}{3} \l_2$, then finite time blow-up cannot
occur. On the other hand we have just seen  
that blow-up occurs in the case $\l_1<\frac{4\pi}{3} \l_2$,
provided the initial energy is sufficiently small, say  
non-positive. What remains open therefore  
is the case of (large) positive initial energy in the unstable
regime. Unfortunately we can only give a partial  
answer to that.

\subsection{Global existence for $d=3$ (unstable regime)}

Here we shall show that global in time existence is possible in the
case  $\l_1<\frac{4\pi}{3} \l_2$ under some additional assumptions. 

\begin{proposition}
Under the same assumptions as in Proposition~\ref{prop:localSigma},
  $d=3$, suppose that in 
  addition $\l_1< \frac{4\pi}{3} \l_2$ and $\l_2\ge 0$.  Then there
  exists a $C_0>0$,  
  independent of $\l_1$ and $\l_2$, and  a (positive) constant 
   \begin{equation*}
    \widetilde C_0 =  \frac{C_0}{M\(\frac{4\pi}{3} \l_2-\l_1\)^2},
    \end{equation*}
  such that, if $ 0<E < \widetilde C_0$ and 
  $\lVert \nabla  \varphi\rVert_{L^2}^2
    < \widetilde C_0$, then $T_*=\infty$. 
\end{proposition}
Of course, this result can only give an additional insight in
situations where $\underline \om^2 \lVert x  \varphi\rVert_{L^2}^2 <
3\widetilde C_0$.  
The picture becomes a bit clearer, though, if one ignores the harmonic
confinement for a moment, i.e. set $\om_j = 0$. Then we know  
that in the unstable regime  $\l_1<\frac{4\pi}{3} \l_2$ blow-up occurs
as soon as the total initial energy is non-positive. On the  
other hand, if the total initial energy is positive but not too large,
i.e. smaller than $\widetilde C_0$,  
then global in time existence still holds, provided the initial
kinetic energy is also smaller than $\widetilde C_0$. 
\begin{proof}
 We resume the same approach as in the proof of Theorem
 \ref{thm:global3DSigma-1}, from which we now get  
 \begin{equation*}
   \lVert \nabla \psi(t)\rVert_{L^2}^2\le 2E +
    \frac{1}{(2\pi)^3}\int_{\R^3} \(\frac{4\pi}{3}\l_2- \l_1  \)\lvert  
    \widehat \rho(\xi)\rvert^2 \D \xi=2E +\(\frac{4\pi}{3}\l_2- \l_1  \) \lVert
    \psi(t)\rVert_{L^4}^4. 
 \end{equation*}
The Gagliardo--Nirenberg inequality then yields
\begin{equation}\label{eq:17h38}
\lVert \nabla \psi(t)\rVert_{L^2}^2\le 2E + C\(\frac{4\pi}{3}\l_2-
    \l_1  \)\lVert 
    \psi(t)\rVert_{L^2} \lVert \nabla \psi(t)\rVert_{L^2}^3.     
\end{equation}
On the other hand, the conservation of mass implies:
\begin{equation*}
  \lVert \nabla \psi(t)\rVert_{L^2}^2\le 2E +
    C\(\frac{4\pi}{3}\l_2- \l_1\)\sqrt M \lVert \nabla \psi(t)\rVert_{L^2}^3.
\end{equation*}
With these two estimates in hand, the result follows from a
bootstrap argument (see e.g. \cite{BG}): 
\begin{lemma}[Bootstrap argument]\label{lem:boot}
Let $f=f(t)$ be a nonnegative continuous function on $[0,T]$ such
that, for every $t\in [0,T]$, 
\begin{equation*}
  f(t)\le \eps_1  + \eps_2 f(t)^\theta,
\end{equation*}
where $\eps_1,\eps_2>0$ and $\theta >1$ are constants such that
\begin{equation*}
  \eps_1 <\left(1-\frac{1}{\theta} \right)\frac{1}{(\theta \eps_2)^{1/(\theta
-1)}}\ ,\ \ \ f(0)\le  \frac{1}{(\theta \eps_2)^{1/(\theta
-1)}}.
\end{equation*}
Then, for every $t\in [0,T]$, we have
\begin{equation*}
  f(t)\le \frac{\theta}{\theta -1}\ \eps_1.
\end{equation*}
\end{lemma}
We apply the lemma with $f(t)=\lVert \nabla \psi(t)\rVert_{L^2}^2$,
$\eps_1=2E$,  $\eps_2 = C\(\frac{4\pi}{3}\l_2- \l_1\)\sqrt M$ and
$\theta=3/2$. Note that we therefore have to assume that the energy is
positive. Otherwise, \eqref{eq:17h38} yields no information any way,
and Theorem~\ref{thm:virial} shows that finite time blow-up
occurs.  
\end{proof}

\section{Dimension reduction}
\label{sec:reduction}

We shall mainly follow the ideas of \cite{BMSW} where such an analysis
has been rigorously performed for $\l_2 =0$ (no dipole effects) in the  
case of modulated ground state intial data. We also remark that the case of
general initial data has been treated in the  
remarkable paper \cite{BCM}. In our case, due to the presence of the
non-isotropic dipole-kernel, we can distinguish two main  
cases: The reduction from $d=3$ to an effective one-dimensional model
in the dipole direction $n=(0,0,1)$, and the reduction to an effective 
two-dimensional model perpendicular to $n$. We shall
sketch the adaptation of the approach in \cite{BMSW} to the present
context. It is very likely that 
the analysis of \cite{BCM} can be adapted to the dipole case, but
we shall not pursue this question as it is beyond the scope of our work.

\subsection{Formal derivation of the one-dimensional model} 
\label{sec:conftodipole}
Let us start with the first problem of deriving an effective
one-dimensional model in the dipole direction. To this end  
we write the general model \eqref{eq:NLSgen} in the form
\begin{equation}\label{eq:conf1}
  i\d_t \psi +\frac{1}{2}\Delta \psi = \frac{1}{\eps^4}\(\om_1^2
  \frac{x_1^2}{2} +\om_2^2
  \frac{x_2^2}{2}\)\psi +\om_3^2\frac{x_3^2}{2}\psi+
  \l_1\lvert\psi\rvert^2\psi +\l_2 \(K\ast 
  \lvert\psi\rvert^2\)\psi.
\end{equation}
The parameter $\eps$ is positive and small, its smallness modeling a
strong confinement in the first two directions. 
Introduce a change of variables via
\begin{equation*}
  \psi(t,x_1,x_2,x_3) =  \psi^\eps
  \(t,\frac{x_1}{\eps},\frac{x_2}{\eps},x_3\). 
\end{equation*}
Then equation~\eqref{eq:conf1} is equivalent to:
\begin{equation}\label{eq:conf2}
  i\d_t  \psi^\eps +\frac{1}{2}\frac{{\d}^2}{\d x_3^2} \, \psi^\eps
  = \frac{1}{\eps^2}H_2 \psi^\eps
  +\om_3^2\frac{x_3^2}{2}\psi^\eps+ 
  \l_1\lvert \psi^\eps\rvert^2 \, \psi^\eps +\l_2
 {\mathcal K}^\eps  \psi^\eps,
\end{equation}
where $H_2$ is the 
two-dimensional harmonic oscillator 
\begin{equation*}
  H_2 = -\frac{1}{2}\(\frac{{\d}^2}{\d x_1^2} +\frac{{\d}^2}{\d
  x_2^2}\) + \om_1^2 
  \frac{x_1^2}{2} +\om_2^2
  \frac{x_2^2}{2},
\end{equation*}
and ${\mathcal K}^\eps $ is the convolution operator
\begin{equation*}
{\mathcal K}^\eps (t,x) = \int_{\R^3} \frac{(\eps x_1-y_1)^2+(\eps
 x_2-y_2)^2-2(x_3-y_3)^2}{\((\eps x_1-y_1)^2+(\eps 
 x_2-y_2)^2+(x_3-y_3)^2\)^{5/2}}  \, \lvert
 \psi^\eps(t,y)\rvert^2 \, \D y.
\end{equation*}
For $\ell \in \N$, we shall denote by $\chi_\ell (x_1,x_2)$ the
(normalized) eigenfunction of $H_2$ corresponding to the  
eigenvalue $\mu_\ell \in \R_+$. In particular the ground state
$\chi_0$ corresponds to  
$\mu_0 = \frac{1}{2}(\om_1+\om_2)$ and is explicitly given by 
\begin{equation*}
  \chi_0(x_1,x_2) = \frac{\sqrt{\om_1\om_2}}{\pi} \, e^{-(\om_1
  x_1^2+\om_2x_2^2)/2}. 
\end{equation*}
As in \cite{BMSW}, we consequently seek a solution to \eqref{eq:conf2} in 
the form
\begin{equation}\label{eq:ansatz}
  \psi^\eps(t,x_1,x_2,x_3) = e^{-i \mu_0
  t/\eps^2}\chi_0(x_1, x_2)u(t,x_3). 
\end{equation}
This type of ansatz obviously requires well-prepared initial data (i.e. concentrated on the ground state of $H_2$).
A formal multiple scales expansion for $ \psi^\eps$ as $ \eps \to 0 $
then yields the consistency relation 
\begin{equation*}
  (H_2  - \mu_0) \psi^\eps=0,
\end{equation*}
plus an evolution equation for the modulation $u(t,x)$, given by
\begin{equation}\label{eq:reduc1D}
  i\d_t u +\frac{1}{2}\frac{{\d}^2}{\d x_3^2} u  = \om_3^2\frac{x_3^2}{2}u+
  \kappa_1\lvert u\rvert^2 u +\l_2 \(K_1\ast 
  \lvert u\rvert^2\)u,
\end{equation}
Here we denote by 
\begin{equation*}
\kappa_1 = \l_1\int_{\R^2} \chi_0^4(x_1, x_2) \, \D x_1 \D x_2,
\end{equation*} 
the effective coupling constant for the cubic nonlinearity and by 
\begin{equation}\label{eq:K1}
  K_1(x_3) = \int_{\R^2} \frac{x_1^2+x_2^2-2x_3^2
  }{\(x_1^2+x_2^2+x_3^2\)^{5/2}} \, \chi_0(x_1,x_2)^2
  \D x_1\D x_2. 
\end{equation}
the corresponding effective one-dimensional dipole kernel. To apply the
existence analysis  
presented in \S\ref{sec:global}, it is important to check if
the operator ${\mathcal K}_1:u\mapsto K_1\ast  
  u$ is bounded on $L^2(\R)$. To do so we shall show
  that the one-dimensional Fourier transform of  
  $K_1$ is in $L^\infty(\R)$ . Recall
  \begin{equation*}
    K(x_1,x_2,x_3)= \frac{x_1^2+x_2^2-2x_3 ^2
  }{\(x_1^2+x_2^2+x_3^2\)^{5/2}}.
  \end{equation*}
Then, denoting by $\F_2$ the partial Fourier transform with respect to the
variables $x_1$ and $x_2$ only, we obtain from Plancherel's formula on
$\R^2$, that 
  \begin{align*}
    \widehat K_1(\xi_3) &=\int_{\R} e^{-ix_3\xi_3}
    K_1(x_3) \, \D x_3 
    \\
&= \int_{\R^3} e^{-ix_3\xi_3}K(x_1,x_2,x_3)
    \overline{\chi_0(x_1,x_2)^2} \, \D x_1\D x_2 \D x_3\\
&= \frac{1}{(2\pi)^2}\int_{\R^3} e^{-ix_3\xi_3} \F_{2}K(\xi_1,\xi_2,x_3)
    \overline{\F_{2} \chi_0^2}(\xi_1,\xi_2) \, \D \xi_1\D \xi_2
    \D x_3\\ 
&= \frac{1}{(2\pi)^2}\int_{\R^2}\overline{\F_{2} \chi_0^2}(\xi_1,\xi_2)
    \, \D \xi_1\D \xi_2\int_{\R}e^{-ix_3\xi_3}
    \F_{2}K(\xi_1,\xi_2,x_3) \, \D x_3\\ 
&= \frac{1}{(2\pi)^2}\int_{\R^3}\overline{\F_{2} \chi_0^2}(\xi_1,\xi_2)
    \widehat K(\xi_1,\xi_2,\xi_3) \, \D \xi_1\D \xi_2 \D \xi_3,
  \end{align*}
where $\widehat K$ denotes the three-dimensional Fourier transform of
the kernel 
$K$. Since $\chi_0$ is a Schwartz function, the Fourier transform of
$\chi_0^2$ is a Schwartz function, thereby in $L^1(\R)$. Consequently, the
boundedness of 
$\widehat K$ on $\R^3$ (which stems from Lemma~\ref{lem:CZ}, or more
explicitly from Lemma~\ref{lem:FourierK}) implies the boundedness of
$\widehat K_1$ on $\R$.  

The assumptions of
Proposition~\ref{prop:localSigma} are thus satisfied. From
Corollary~\ref{cor:global1DSigma}, we know that the effective
one-dimensional model \eqref{eq:reduc1D} has a  
unique, global solution in $\Sigma$. Moreover, the same holds true if
$\Sigma$ is 
replaced by $L^2(\R)$, as proved in Theorem~\ref{thm:globalL2}.

\subsection{Formal derivation of the two-dimensional model}
\label{sec:confdirdipole}
We proceed as before. Starting from 
\begin{equation}\label{eq:conf12}
  i\d_t \psi +\frac{1}{2}\Delta \psi = \(\om_1^2
  \frac{x_1^2}{2} +\om_2^2
  \frac{x_2^2}{2}\)\psi +\frac{\om_3^2}{\eps^4}\frac{x_3^2}{2}\psi+
  \l_1\lvert\psi\rvert^2\psi +\l_2 \(K\ast 
  \lvert\psi\rvert^2\)\psi,
\end{equation}
we introduce a
new change of variables by  
\begin{equation*}
  \psi(t,x_1,x_2,x_3) =  \widetilde \psi^\eps
  \(t,x_1,x_2,\frac{x_3}{\eps}\). 
\end{equation*}
Equation~\eqref{eq:conf12} then becomes
\begin{equation*}
  i\d_t  \widetilde \psi^\eps +\frac{1}{2}\(\frac{{\d}^2}{\d x_1^2}
  +\frac{{\d}^2}{\d x_2^2}\) \widetilde \psi^\eps 
  = \frac{1}{\eps^2}H_1 \widetilde \psi^\eps
  +\frac{1}{2} \left(\om_1^2 x_1^2 + \om_2^2 x_2^2 \right)\widetilde
  \psi^\eps+  
  \l_1\lvert \widetilde \psi^\eps\rvert^2 \, \widetilde \psi^\eps +\l_2
 \widetilde {\mathcal K}^\eps  \widetilde \psi^\eps,
\end{equation*}
where $H_1$ now denotes the one-dimensional harmonic oscillator,
acting in the dipole direction 
\begin{equation*}
  H_1 = -\frac{{\d}^2}{\d x_3^2} + \om_3^2
  \frac{x_3^2}{2}.
\end{equation*}
Denoting the corresponding ground state by $\tilde \chi_0(x_3)$ and
proceeding analogously as before we arrive,  
instead of \eqref{eq:conf2}, at the 
following effective equation for the modulation $\widetilde u(t,x_1,x_2)$: 
\begin{equation}\label{eq:reduc2D}
  i\d_t  \widetilde u +\frac{1}{2}\(\frac{{\d}^2}{\d x_1^2}
  +\frac{{\d}^2}{\d x_2^2}\) \widetilde u 
  = \frac{1}{2} \left(\om_1^2 x_1^2 + \om_2^2 x_2^2 \right)\widetilde u+ 
  \widetilde \kappa_1\lvert \widetilde u \rvert^2 \, \widetilde u +\l_2
 K_2  \widetilde u,
\end{equation}
where now
\begin{equation*}
\widetilde \kappa_1 = \l_1\int_{\R} \widetilde \chi_0^4(x_3) \, \D x_3,
\end{equation*} 
is the new coupling constant and
\begin{equation*}
  K_2(x_1,x_2)= \int_{\R} \frac{x_1^2+x_2^2-2x_3^2
  }{\(x_1^2+x_2^2+x_3^2\)^{5/2}} \, \chi_0(x_3)^2 \, \D x_3.
\end{equation*}
is the effective two-dimensional dipole kernel. 
By the same computation as above, we see that $\F_2 K_2 \in
L^\infty(\R^2)$, and thus the assumptions of
Proposition~\ref{prop:localSigma} are satisfied to guarantee a well-posed 
initial value problem in $\Sigma$, at least 
locally in time.

\subsection{A rigorous result} 

For completeness we shall finally state a rigorous mathematical 
result and sketch the corresponding proof which 
follows the lines of \cite{BMSW}. We only consider the case $\l_1\ge
\frac{4\pi}{3} \l_2\ge 0$  
(the stable regime), 
which allows for a somewhat shorter argument and again refer to
\cite{BMSW, BCM} for more general statements.  
Moreover, for notational convenience, we shall state the result only
for the case of   
\eqref{eq:reduc1D} (one-dimensional model in dipole direction), 
but exactly the same approach can be followed for justifying
\eqref{eq:reduc2D}. We consider initial data for $\psi^\eps$ which do
not depend on $\eps$, and refer to \cite{BMSW, BCM} for a discussion
on this assumption. 
\begin{theorem}  Let $d=3$, $V_d$ be quadratic and $\l_1\ge
  \frac{4\pi}{3} \l_2\ge 0$. Denote by $\psi^\eps(t) \in \Sigma$  
and $u(t)\in \Sigma$ the unique solutions of \eqref{eq:conf2} and
\eqref{eq:reduc1D}, respectively, with
\begin{equation*}
  \psi^\eps\big |_{t=0}= \chi_0(x_1,x_2)u_0(x_3)\ ,\quad
  u\big |_{t=0}=u_0(x_3), \quad u_0\in \Sigma .
\end{equation*}
Then, for any $T < \infty$, there exists $C_T>0$ such that 
\begin{equation*}
\sup_{t \in [0,T]} \left \|  \psi^ \eps(t) -e^{-i \mu_0 t/
    \eps^2}u(t) \chi_0 \right \|_{L^2(\R^3)}  \le C_T \eps . 
\end{equation*}
\end{theorem}
\begin{proof}[Proof (Sketch)]
The energy associated to \eqref{eq:conf2} can be written as
\begin{equation*}
\begin{split}
  E=& \< \psi^\eps(t), H_1 \psi^\eps(t) \> + \frac{1}{\eps^2}
  \< \psi^\eps(t), H_2 \psi^\eps(t) \> \\ 
  &+ \frac{1}{2(2\pi)^3}  \int_{\R^3}\(\l_1 +\frac{4\pi}{3} \l_2
\(3\frac{\xi_3^2}{\lvert 
  \xi\rvert^2}-1\) \)\lvert\widehat
\rho(\xi)\rvert^2\D \xi,
\end{split}
\end{equation*}
where  $\langle \cdot, \cdot \rangle$ denotes the scalar product in
$L^2(\R^3)$, and 
$H_1$, $H_2$  denote the one-dimensional and
two-dimensional harmonic oscillator operator, respectively.  
In the following we denote
by $\Pi_\ell$ the orthogonal projector onto the eigenspace
corresponding to $\chi_\ell(x_3)$ and define  
\begin{equation*}
\psi_\ell^\eps(t,x_3) =  e^{i \mu_0 t / \eps^2} \int_{\R^2}
\psi^\eps (t,x_1,x_2,x_3) \chi_\ell (x_1,x_2) dx_1dx_2.  
\end{equation*}
Recalling the ansatz \eqref{eq:ansatz} for $\psi^\eps(t)$ and using
the fact that the  
eigenfunctions $\{ \chi_\ell \}_{\ell \in \N}$ form an orthonormal
basis of $L^2(\R^2)$ we note that  
$\chi_{0\mid t=0}^\eps = u_{\mid t=0}$, since $\| \chi_0 \|_{L^2}
=1$. Thus, we can rewrite 
\begin{equation*}
 \left < \psi^\eps(t), H_2 \psi^\eps(t) \right > = \sum_{\ell
 =1}^\infty (\mu_\ell - \mu_0) \| \psi_\ell^\eps(t) \|_{L^2}^2 +   
 \mu_0 \| u(0) \|_{L^2}^2.
\end{equation*}
Keeping in mind that, by assumption $\left < \psi^\eps(0), H_2
  \psi^\eps(0) \right > =  \mu_0 \| u(0) \|_{L^2(\R)}^2$, the
conservation of energy implies 
\begin{equation*}
\begin{split}
 E_{\rm red}=& \left < \psi^\eps(t), H_1 \psi^\eps(t) \right > +
  \frac{1}{\eps^2} \sum_{\ell =1}^\infty (\mu_\ell - \mu_0) \|
  \psi_\ell^\eps (t) \|_{L^2}^2 \\ 
  &+ \frac{1}{2(2\pi)^3}  \int_{\R^3}\(\l_1 +\frac{4\pi}{3} \l_2
\(3\frac{\xi_3^2}{\lvert 
  \xi\rvert^2}-1\) \)\lvert\widehat
\rho(\xi)\rvert^2\D \xi,
\end{split}
\end{equation*}
where the reduced initial energy is
\begin{equation*}
\begin{split}
  E_{\rm red}= & \, \left < \psi^\eps(0), H_1 \psi^\eps(0) \right > +
  \frac{1}{2(2\pi)^3}  \int_{\R^3}\(\l_1 +\frac{4\pi}{3} \l_2 
\(3\frac{\xi_3^2}{\lvert   \xi\rvert^2}-1\) \)\lvert\widehat
  \rho(\xi)\rvert^2\D \xi \\ 
\ge & \, \left < \psi^\eps(0), H_1 \psi^\eps(0) \right > +
  \frac{1}{2(2\pi)^3} \int_{\R^3} \(\l_1 -\frac{4\pi}{3} \l_2 
    \)\lvert     \widehat \rho(\xi)\rvert^2 \D \xi\ge 0, \\
    \end{split}
\end{equation*}
by assumption on the parameters $\l_1, \l_2$. On the other hand, by
exactly the same argument we infer from the energy conservation above
that  
\begin{align*}
  E_{\rm red} \ge  & \, \left < \psi^\eps(t), H_1 \psi^\eps(t) \right
  > + \frac{1}{\eps^2} \sum_{\ell =1}^\infty (\mu_\ell - \mu_0) \|
  \psi_\ell^\eps(t) \|_{L^2}^2 \\ 
  & \, + \frac{1}{2(2\pi)^3} \int_{\R^3} \(\l_1 -\frac{4\pi}{3} \l_2
    \)   \lvert \widehat \rho(\xi)\rvert^2 \D \xi.
\end{align*}
All the terms in this equation are non-negative and thus, we
immediately obtain 
\begin{equation*}
\sum_{\ell =1}^\infty (\mu_\ell - \mu_0) \| \psi_\ell^\eps(t) \|_{L^2}^2
\le C \eps^2. 
\end{equation*}
Denoting by $\Pi_0$ the projection onto the eigenspace generated by
the ground state $\chi_0$, we can thus estimate, since $\psi_0^\eps
\chi_0 = e^{i\mu_0t/\eps^2}\Pi_0 \psi^\eps$, 
\begin{align*}
& \left \| \psi^ \eps(t) -  e^{-i \mu_0 t/ \eps^2}  \chi_0 u(t)
 \right \|_{L^2(\R^3)} \\ 
 & \, \le \left \| (1- \Pi_0)  \psi^ \eps(t) \right \|_{L^2(\R^3)}   +
 \left \|  e^{-i \mu_0 t/ \eps^2}\chi_0 (\psi_0^\eps(t) -u(t))
 \right \|_{L^2(\R^3)}  \\ 
 & \, = \left \| (1- \Pi_0)  \psi^ \eps(t) \right \|_{L^2(\R^3)}   + 
\left \|  \psi_0^\eps(t) -u(t)\right \|_{L^2(\R)}.
\end{align*}
For the first term on the right hand side, we have
\begin{equation*}
 \left \| (1- \Pi_0)  \psi^ \eps(t) \right \|_{L^2(\R^3)}^2
 =\sum_{\ell =1}^\infty \left\| \psi_\ell^\eps(t)\right\|_{L^2}^2 
 \le \frac{1}{\mu_1 - \mu_0} \sum_{\ell =1}^\infty (\mu_\ell - \mu_0)
 \| \psi_\ell^\eps(t) \|_{L^2}^2 \le C \eps^2. 
 \end{equation*}
Applying one derivative with respect to $x_1$ or $x_2$, Sobolev
embedding shows that the above estimate remains valid if the
$L^2(\R^3)$-norm is replaced by $L^2_{x_3}\(L^p_{x_1,x_2}\)$, for
$2\le p<\infty$. 

For the term $\psi_0^\eps-u$, more computations are necessary, and we shall
emphasize the main arguments only. First, as in \cite{BMSW}, we check that
$\psi_0^\eps$ solves
\begin{equation*}
  i\d_t \psi_0^\eps = H_1 \psi_0^\eps + F_1^\eps,
\end{equation*}
where 
\begin{equation*}
  F_1^\eps(t,x_3) = e^{i\mu_0 t/\eps^2}\int_{\R^2} \(\l_1 \lvert
  \psi^\eps\rvert^2 \psi^\eps +\l_2 {\mathcal K}^\eps\psi^\eps\)
  \chi_0(x_1,x_2)dx_1dx_2. 
\end{equation*}
Recall that $u$ solves
\begin{equation*}
  i\d_t u = H_1 u + \l_1 \lvert u\rvert^2 u \int_{\R^2}\chi_0^4
  (x_1,x_2)dx_1dx_2 +\l_2 \(K_1\ast \lvert u\rvert^2\) u, 
\end{equation*}
where $K_1$ is given by \eqref{eq:K1}. Introducing the error $w^\eps =
\psi_0^\eps -u$, we can write
\begin{equation*}
  i\d_t w^\eps = H_1 w^\eps + L^\eps + \l_1 S_1^\eps+\l_2 S_2^\eps,
\end{equation*}
where
\begin{align*}
  L^\eps&= \l_1 \(\lvert \psi_0^\eps\rvert^2 \psi_0^\eps - \lvert
  u\rvert^2 u\)\int_{\R^2}\chi_0^4
   +\l_2 \(  \(K_1\ast \lvert \psi_0^\eps\rvert^2\)
  \chi_0^\eps-\(K_1\ast \lvert u\rvert^2\) u\), \\ 
S_1^\eps &= e^{i\mu_0 t/\eps^2}\(\Pi_0\(\lvert \psi^\eps\rvert^2
  \psi^\eps\) - \lvert \Pi_0\psi^\eps\rvert^2
  \Pi_0\psi^\eps\),\\
S_2^\eps&= e^{i\mu_0 t/\eps^2}\Pi_0\(   {\mathcal K}^\eps\psi^\eps
  \) - \(K_1\ast \lvert
  \psi_0^\eps\rvert^2\)\psi_0^\eps  .
\end{align*}
The end of the proof relies essentially on the same estimates as in the
contraction part of the  proof of Proposition~\ref{prop:localSigma} in
the case $d=1$,
based on Strichartz estimates; we refer to  \cite{BMSW} for
adaptations due to the fact that the variables $x_1$ and $x_2$ appear
as parameters here. Roughly speaking, the term $L^\eps$ is
treated like a linear term in a Gronwall lemma, even though this not so clear unless one goes 
through the details of computation (which we shall not do):  for tiny
time intervals, the term corresponding to $L^\eps$ is ``absorbed by
the left hand side''. Since we consider only finite time intervals, we
proceed this way a finite number of times, so the size of $w^\eps$ is
dictated by the size of the source terms $S_1^\eps$ and $S_2^\eps$ in
suitable mixed time-space norms. 

In view of the above remark on the
estimate for $(1-\Pi_0)\psi^\eps$, $S_1^\eps$ is $\O(\eps)$ in the
spaces we need; this point is exactly the same as in \cite{BMSW}, so we
leave it out. To conclude, we simply consider the case of
$S_2^\eps$. Recall
 \begin{equation*}
{\mathcal K}^\eps = \int_{\R^3}
 \frac{(\eps x_1-y_1)^2+(\eps 
 x_2-y_2)^2-2(x_3-y_3)^2}{\((\eps x_1-y_1)^2+(\eps 
 x_2-y_2)^2+(x_3-y_3)^2\)^{5/2}}   \lvert
 \psi^\eps(t,y)\rvert^2 d y. 
\end{equation*}
Write
\begin{equation*}
  \lvert \psi^\eps(t,y)\rvert^2 =\lvert \psi^\eps(t,y)\rvert^2  -
  \lvert \Pi_0 \psi^\eps(t,y)\rvert^2 + \lvert \Pi_0
  \psi^\eps(t,y)\rvert^2 .   
\end{equation*}
Up to an error of order $\O(\eps)$, we can replace ${\mathcal K}^\eps$
by
\begin{equation*}
{\mathcal K}_1^\eps=\int_{\R^3}
 \frac{(\eps x_1-y_1)^2+(\eps 
 x_2-y_2)^2-2(x_3-y_3)^2}{\((\eps x_1-y_1)^2+(\eps 
 x_2-y_2)^2+(x_3-y_3)^2\)^{5/2}}   \lvert
 \psi_0^\eps(t,y_3)\chi_0(y_1,y_2)\rvert^2 d y  
\end{equation*}
Essentially, the only new term to estimate is
\begin{equation*}
  \Pi_0 {\mathcal K}_1^\eps - K_1\ast \lvert \Pi_0
  \psi^\eps\rvert^2 = 
\Pi_0 {\mathcal K}_1^\eps - K_1\ast \lvert \psi_0^\eps \chi_0\rvert^2. 
\end{equation*}
Up to an extra ``Gronwall term'', we can replace this by
$\Pi_0 \widetilde{\mathcal K}_1^\eps - K_1\ast \lvert u
\chi_0\rvert^2$, where
 \begin{equation*}
\widetilde {\mathcal K}_1^\eps=\int_{\R^3}
 \frac{(\eps x_1-y_1)^2+(\eps 
 x_2-y_2)^2-2(x_3-y_3)^2}{\((\eps x_1-y_1)^2+(\eps 
 x_2-y_2)^2+(x_3-y_3)^2\)^{5/2}}   \lvert
 u(t,y_3)\chi_0(y_1,y_2)\rvert^2 d y  
\end{equation*}
The difference $\Pi_0 \widetilde{\mathcal K}_1^\eps - K_1\ast \lvert u
\varphi\rvert^2=:I^\eps(t,x_3)$ reads
\begin{equation*}
 I^\eps= \int_{\R^5} \delta^\eps (x_1,x_2,x_3,y_1,y_2,y_3)
  |u(t,y_3)|^2\chi_0(y_1,y_2)^2\chi_0(x_1,x_2)^2dy_1dy_2dy_3dx_1
  dx_2,  
\end{equation*}
where
\begin{align*}
 \delta^\eps &= K\(y_1-\eps x_1,y_2-\eps x_2,y_3-x_3\) -K\(y_1,y_2,y_3-x_3\)\\
 &= \frac{(\eps x_1-y_1)^2+(\eps 
 x_2-y_2)^2-2(x_3-y_3)^2}{\((\eps x_1-y_1)^2+(\eps 
 x_2-y_2)^2+(x_3-y_3)^2\)^{5/2}} - \frac{y_1^2+y_2^2-2(x_3-y_3)^2 
  }{\(y_1^2+y_2^2+(x_3-y_3)^2\)^{5/2}}.  
\end{align*}
Plancherel's formula with respect to the variables $y_1$
and $y_2$ yields:
\begin{align*}
  I^\eps(t,x_3) = \frac{1}{(2\pi)^2}
\int_{\R^5}& \(e^{-i\eps(x_1 \eta_1+x_2\eta_2)}-1\)\F_2
  K(\eta_1,\eta_2,y_3-x_3)
  \overline{\F_2\(\chi_0^2\)}(\eta_1,\eta_2) \\
&\times|u(t,y_3)|^2\chi_0(x_1,x_2)^2
  d\eta_1d\eta_2 dy_3dx_1dx_2.
\end{align*}
Now computing the $L^2$ norm of $I^\eps$ (as a function of $x_3$), the
boundedness of the three-dimensional Fourier transform of $K$ and the
boundedness of $u$ in $L^2(\R)$ yield
\begin{equation*}
  \| I^\eps(t)\|_{L^2}^2\le C \int_{\R^4} \left\lvert e^{-i\eps(x_1
  \eta_1+x_2\eta_2)}-1\right\rvert^2 \left\lvert
  \F_2\(\chi_0^2\)(\eta_1,\eta_2)\chi_0(x_1,x_2)^2\right\rvert^2  
  d\eta_1d\eta_2dx_1dx_2. 
\end{equation*}
Recalling that $\chi_0\in \Sch(\R^2)$, the standard estimate
$|e^{i\theta}-1|\le |\theta|$ for $\theta\in \R$ shows that the $L^2$
norm of $I^\eps$ is $\O(\eps)$, uniformly for $t\ge 0$. Therefore, its
$L^1([0,T];L^2)$ norm (the norm that appears in energy estimates,
which are a particular case of Strichartz estimates) 
is of order $\O(\eps T)$, and the proposition follows. 
\end{proof}

\bibliographystyle{amsplain}
\bibliography{biblio}

\end{document}